\documentclass{article}
\usepackage[utf8]{inputenc}
\usepackage[a4paper, total={6.5in, 9.5in}]{geometry}
\usepackage{amsthm,amsmath,amssymb}
\usepackage{charter,eulervm} 
\usepackage{algorithm}
\usepackage[noend]{algpseudocode}
\usepackage{authblk}
\usepackage{url}

\newcommand{\midset}{{\rm mid}}
\newcommand{\lw}{{\rm lw}}
\newcommand{\pw}{{\rm pw}}

\newtheorem{lemma}{Lemma}
\newtheorem{theorem}{Theorem}

\title{A Note on Exponential-Time Algorithms for Linearwidth}
\author{Yasuaki Kobayashi}
\author{Yu Nakahata}
\affil{Kyoto University}
\date{}

\begin{document}

\maketitle

\begin{abstract}
    In this note, we give an algorithm that computes the linearwith of  input $n$-vertex graphs in time $O^*(2^n)$, which improves a trivial $O^*(2^m)$-time algorithm, where $n$ and $m$ the number of vertices and edges, respectively.
\end{abstract}

\section{Introduction}
{\em Width parameters} of graphs play indispensable roles in many graph algorithms. 
{\em Treewidth} and {\em Branchwidth} measure how graphs can be finely decomposed into tree-like structures, and if a graph has small treewidth (or branchwidth), many NP-hard problems can be
solved efficiently on this graph using underlying tree-like structures.

While treewidth and branchwidth are defined through tree-structured decompositions, {\em Pathwidth} and {\em linearwidth}, which are ``linear variants'' of these parameters, are defined through path-structured decompositions.
Many NP-hard graph problems can be solved efficiently on graphs having small those width parameters as well as treewidth and branchwidth.
In this note, we focus on linearwidth.
The linearwidth of graphs was mentioned for the first time in the lecture note given by Robin Thomas\footnote{ \url{https://people.math.gatech.edu/~thomas/tree.ps}}.
This parameter is known as a ``linear counterpart'' of branchwidth.
Thilikos pointed out some applications to several graph searching algorithms~\cite{Thilikos:algorithms:2000}.
As for algorithmic applications, linearwidth is frequently used in enumerating subgraphs with specific properties with the frontier-based methods on Zero-suppressed binary Decision Diagrams (ZDD) and is an important measure for the efficiency of those enumeration algorithms \cite{Kawahara:Frontier:2017,Knuth:TAOCP:2011,Sekine:Computing:1995}. 

The problems of computing the treewidth, branchwidth, pathwidth, and linearwidth of graphs are all NP-hard \cite{Arnborg:Complexity:1987,Seymour:Call:1994,Kashiwabara:NP-compeleteness:1979,Thilikos:algorithms:2000}.
From the viewpoint of exact exponential-time algorithms, there are several attempts beating trivial brute force search algorithms for computing these width parameters: There are exact algorithms for treewidth, branchwidth, pathwidth for graphs having $n$ vertices and $m$ edges that run in time $O(1.7549^n)$~\cite{Fomin:Exact:2008,Fomin:Treewidth:2012}, $\min\{O^*(2^m), O^*(2\sqrt{3})^n\}$~\cite{Fomin:Computing:2009,Oum:Computing:2009}, and $O(1.89^n)$~\cite{Kitsunai:Computing:2016}, respectively, where the $O^*$ notation suppresses the polynomial factor in the input size.
However, to the best of authors' knowledge, the best known exact algorithm for linearwidth is a Bellman-Held-Karp style dynamic programming algorithm, which runs in time $O^*(2^m)$.

In this note, we give an improved algorithm for linearwidth.

\begin{theorem}\label{thm:main}
    There is an algorithm that computes the linearwidth of an input graph $G$ in time $O^*(2^n)$, where $n$ is the number of vertices in $G$.
\end{theorem}

Fomin and Thilikos~\cite{Fomin:3-approximation:2006} claimed that the pathwidth and linearwidth of every graph differ by at most one.
More precisely, they claimed that the linearwidth of a graph is at least its pathwidth and at most its pathwidth plus one.
However, there are some exceptions on this relation: the linearwidth of the complete graph of two vertices is zero and its pathwidth is one.
In fact, these extreme cases are only exceptions for their inequalities, and we give a corrected proof in Section~\ref{sec:preli}.
As a straightforward application of this relation, together with the exact exponential-time algorithm for pathwidth~\cite{Kitsunai:Computing:2016}, we have the following exponential-time approximation.

\begin{theorem}\label{thm:approx}
    There is an $O(1.89^n)$-time approximation algorithm for linearwidth with additive error of at most one.
\end{theorem}

\section{Preliminaries}\label{sec:preli}
Let $G = (V, E)$ be a graph.
We use $n$ and $m$ to denote the number of vertices and edges of $G$, respectively.
For a subset $F$ of $E$, $\midset(F)$ is the set of vertices having an incidental edge in both $F$ and $E \setminus F$, that is, $\midset(F) = V(F) \cap V(E \setminus F)$, where $V(F)$ is the set of end vertices of edges in $F$. 
We let $d(F) = |\midset(F)|$.

\begin{lemma}\label{lem:submodular}
    The function $d$ is submodular on $E$, that is, for $X, Y \subseteq E$ with $X \subseteq Y$,
    and $e \in E \setminus Y$,
    \[
        d(X \cup \{e\}) - d(X) \ge d(Y \cup \{e\}) - d(Y).
    \]
\end{lemma}
\begin{proof}
    Let $u$ be one of the end vertices of $e$. 
    If the degree of $u$ is exactly one, then it does not contribute $d(F)$ for any $F \subseteq E$. 
    Thus, we assume that the degree of $u$ is at least two.
    For $F \subseteq E$, we let $d_F(u)$ be the number of edges in $F$ incident to $u$.
    Observe that $u$ contributes one to $d(F \cup \{e\}) - d(F)$ if $d_F(u) = 0$, minus one if $d_F(u) = d_E(u) - 1$, and zero otherwise.
    If $d_Y(u) = 0$, then $d_X(u) = 0$ and hence $u$ contributes one to both sides.
    If $d_Y(u) = d_E(u) - 1$, then $u$ contributes minus one to the right-hand side and at least minus one to the left-hand side.
    Finally, if $0 < d_Y(u) < d_E(u) - 1$, then $u$ contributes zero to the right-hand side and at least zero as $d_X(u) \le d_Y(u) < d_E(u) - 1$.
    By symmetrically considering the other end vertex of $e$, the lemma follows.
\end{proof}


For a non-negative integer $k$, let $[k] = \{1, 2, \ldots, k\}$.
A bijection $\pi: [m] \to E$ is called a {\em layout} of $G$.
For $1 \le i \le m$, we denote by $\pi_{\le i}$ the set of edges appeared in the first $i$ edges of $\pi$, i.e., $\pi_{\le i} = \{\pi(j) : 1 \le j \le i\}$.
The {\em width} of $\pi$ is defined as $\max_{1 \le i \le m}d(\pi_{\le i})$ and the {\em linear-width} of $G$ is the minimum integer $k$ such that $G$ has a layout of width at most $k$.

A sequence $\mathcal X = (X_1, X_2, \ldots, X_t)$ of subsets of $V$ is called a {\em path decomposition} of $G$ if the following conditions hold:
\begin{itemize}
    \item $\bigcup_{1\le i \le t}X_i = V$;
    \item for each $e \in E$, there is an index $i$ with $e \subseteq X_i$; and
    \item for each $v \in V$, the indices of bags containing $v$ are consecutive on $\mathcal X$.
\end{itemize}
The {\em width} of $\mathcal X$ is defined as $\max_{1\le i \le t}|X_i| - 1$ and the {\em pathwidth} of $G$ is the minimum integer $k$ such that $G$ has a path decomposition of width at most $k$.

The linear-width and pathwidth of a graph are closely related to each other. Fomin and Thilikos \cite{Fomin:3-approximation:2006} showed that the linear-width and pathwidth differ by at most one for every graph.
More precisely, they claimed that $\pw(G) \le \lw(G) \le \pw(G) + 1$ for every graph $G$.
Although these bounds are almost true, there are some exceptions for this relation: For example, $\pw(K_2) = 1$ and $\lw(K_2) = 0$, where $K_2$ is the complete graph of two vertices.
The following lemma correctly handles such exceptions and its proof is essentially the same as one in \cite{Fomin:3-approximation:2006}.

\begin{lemma}\label{lem:pw-lw}
    For every graph $G$ with $\lw(G) \ge 1$,
    \[
        \pw(G) \le \lw(G) \le \pw(G) + 1.
    \]
\end{lemma}
\begin{proof}
    The bound $\lw(G) \le \pw(G) + 1$ was given in \cite{Fomin:3-approximation:2006}.
    Therefore, we prove here that $\pw(G) \le \lw(G)$.
    
    If $G$ has two or more connected components, then $\pw(G)$ and $\lw(G)$ correspond to the maximum pathwidth and linearwidth of its components, respectively.
    Thus, we assume that $G$ is connected.
    Moreover, it is easy to see that $\lw(G) = 0$ if and only if $G$ is either an isolated vertex or a single edge.
    Therefore, we assume otherwise.
    
    Let $\pi$ be a layout of $G$ of width $\lw(G)$.
    Now, we construct a path decomposition $\mathcal X = (X_1, X_2, \ldots, X_m)$ as $X_i = \midset(\pi_{i - 1}) \cup \pi(i)$.
    In the following, we prove that $\mathcal X$ is a path decomposition of $G$ and its width is at most $\lw(G)$.
    By the definition of layout, the first and second conditions of path decompositions clearly hold.
    Fix a vertex $v \in V$.
    If $i$ is the smallest index with $v \in \pi(i)$, then we have $v \in X_i$.
    Let $j$ be the largest index with $v \in \pi(j)$.
    Then, $v \in \midset(\pi_{\le k})$ for every $i \le k \le j$ and hence the third condition of path decompositions holds.
    
    To bound the width of path decomposition $\mathcal X$, consider a set $X_i$ in $\mathcal X$.
    As $d(\pi_{\le i}) \le \lw(G)$, $\midset(\pi_{\le i-1}) \cap \pi(i) \neq \emptyset$ implies $|X_i| \le \lw(G) + 1$.
    Therefore, we suppose otherwise $\midset(\pi_{\le i-1}) \cap \pi(i) = \emptyset$.
    Since $G$ is connected and is not the complete graph with at most two vertices, at least one of two end vertices of $\pi(i)$ has degree more than one.
    This means that $d(\pi_{\le i-1}) < d(\pi_{\le i}) \le \lw(G)$.
    Therefore, we have $|X_i| \le d(\pi_{\le i-1}) + |\pi(i)| \le \lw(G) + 1$.
\end{proof}

We can compute the pathwidth of a graph in time $O(1.89^n)$~\cite{Kitsunai:Computing:2016} and if we have a path decomposition of width at most $k$, we can compute a layout of $G$ of width at most $k + 1$ in polynomial time~\cite{Fomin:3-approximation:2006}.
It is easy to observe that $\lw(G) = 0$ if and only if each connected component of $G$ has at most two vertices.
As a consequence of these results, we have Theorem~\ref{thm:approx}.

\section{An $O^*(2^n)$-time exact algorithm}
Our proposed algorithm is based on Bellman-Held-Karp style dynamic programming for ordering problems~\cite{Bodlaender:Note:2012}.
Let $F \subseteq E$ and let $k$ be a non-negative integer.
A {\em partial layout} of $F$ is a bijection $\sigma: [|F|] \to F$, and we say that $\sigma$ has width at most $k$ if for every $1\le i \le |F|$, it holds that $d(\sigma_{\le i}) \le k$, where $\sigma_{\le i}$ is defined analogously.
For each $F \subseteq E$, our algorithm determines if there is a partial layout of $F$ of width at most $k$.
Suppose $F$ is not empty.
It is easy to observe that $F$ has a partial layout of width at most $k$ if and only if $d(F) \le k$ and $F \setminus \{e\}$ has a partial layout of width at most $k$ for some $e \in F$.
This allows us to compute the linear-width of $F$ in time $O^*(2^m)$ using dynamic programming over all subsets of $E$.

To obtain a better running time, we use a pruned dynamic programming algorithm based on the $O^*(2^m)$-time algorithm described above.
Let $\sigma$ be a partial layout of $F$ of width at most $k$.
We say that $\sigma$ is {\em $k$-extendable} if $G$ has a layout $\pi$ of width at most $k$ such that $\sigma(i) = \pi(i)$ for every $1 \le i \le |F|$.
The following lemma is the key to our pruned dynamic programming.

\begin{lemma}\label{lem:prune}
    Let $\sigma$ be a partial layout of $F$ of width at most $k$.
    Suppose there is $e \in E \setminus F$ such that $d(F) \ge d(F \cup \{e\})$.
    Then, $\sigma$ is $k$-extendable if and only if the partial layout $\sigma'$ of $F \cup \{e\}$ obtained from $\sigma$ by appending $e$ at the end of $\sigma$ is $k$-extendable.
\end{lemma}

\begin{proof}
    By the definition of $k$-extendability, if $\sigma'$ is $k$-extendable, then so is $\sigma$. Thus, we consider the converse direction.
    
    Let $\sigma$ be a $k$-extendable partial layout of $F$.
    Then, there is a layout $\pi$ of $G$ such that $\pi(i) = \sigma(i)$ for $1 \le i \le |F|$.
    Let $j = |F|$ and let $j'$ be such that $\pi(j') = e$.
    If $j + 1 = j'$, we are done.
    Thus we suppose $j + 1 < j'$.
    Let $\pi'$ be the layout of $G$ obtained from $\pi$ by moving $e$ forward to the $j+1$-th position, that is,
    \[
        \pi' = \pi(1) \cdots \pi(j)\ e\ \pi(j+1)\cdots \pi(j'-1)\ \pi(j'+1) \cdots \pi(m).
    \]
    To prove the $k$-extendability of $\sigma'$, we show that $d(\pi'_{\le i}) \le k$ for every $1 \le i \le m$.
    For each $1 \le i \le j$, $\pi(i) = \pi'(i)$ and for each $j' + 1 \le i \le m$, $\pi_{\le i} = \pi'_{\le i}$.
    Thus, we consider $j < i \le j'$.
    As $\pi'_{\le i} = \pi_{\le i-1} \cup \{e\}$, by the submodularity of $d$, we have
    \begin{align*}
            d(\pi'_{\le i}) - d(\pi_{\le i}) &= d(\pi_{\le i} \cup \{e\}) - d(\pi_{\le i}) \\
            & \le d(F \cup \{e\}) - d(F).
    \end{align*}
    Thus, $d(\pi'_{\le i}) \le d(F \cup \{e\}) - d(F) + d(\pi_{\le k})$, and as $d(F) \ge d(F \cup \{e\})$ and $d(\pi_{\le i}) \le k$, we have $d(\pi'_{\le i}) \le k$.
\end{proof}

This lemma is in fact a special case of the commitment lemma in \cite{Kitsunai:Computing:2016}, and used also for experimental speedup in pathwidth computation \cite{Kobayashi:Search:2014}.
Since in the proof of Lemma~\ref{lem:prune} the property of function $d$ needed is only the submodularity, the lemma also holds any measures on linear layouts satisfying the submodularity, such as pathwidth.

Let $F \subseteq E$.
The {\em closure} of $F$, denoted by $F^*$, is the set of edges contained in the subgraph of $G$ induced by $V(F)$.
By the definition of closure, we have $V(F \cup \{e\}) = V(F)$ for every $e \in F^* \setminus F$.
In particular, we have $\midset(F \cup \{e\}) \subseteq \midset(F)$.
Therefore, by Lemma~\ref{lem:prune}, it suffices to run our dynamic programming over all the closures, that is, $F \subseteq E$ with $F^* = F$.
The pseudocode is described in Algorithm~\ref{alg:main}.

\begin{algorithm}[ht]
\caption{Does $F$ have a $k$-extendable layout?}\label{alg:main}
\begin{algorithmic}[1]

\Procedure{extendable}{$k, F \subseteq E$}
    \If{$F = E$}
        \State \Return \bf{True}
    \EndIf
    \If{$V(F)$ is in the table}
        \State \Return \bf{False}
    \EndIf
    \For{$e \in E \setminus F$ with $d(F \cup \{e\}) \le k$}
        \If{{\sc isextendable}($k$, $(F \cup \{e\})^*$)}
            \State \Return \bf{True}
        \EndIf
    \EndFor
    \State Store $V(F)$ in the table.
    \State \Return \bf{False}
\EndProcedure

\end{algorithmic}
\end{algorithm}

The set of edges in the closure $F^*$ of $F$ is exactly the edges in the subgraph induced by $V(F)$.
Therefore, the number of distinct closures is at most $2^n$, which proves Theorem~\ref{thm:main}.

\bibliographystyle{plain}
\bibliography{main}

\end{document}